\documentclass{article}
\usepackage{spconf}
\usepackage{psfrag}
\usepackage{color}
\usepackage[compress]{cite}
\usepackage[pdftex]{graphicx}
\usepackage{amssymb}
\usepackage[cmex10]{amsmath}
\usepackage{amsthm}
\usepackage{algorithm, algorithmic}
\usepackage{array,enumerate}
\usepackage{multirow}
\usepackage[tight, footnotesize]{subfigure}

\graphicspath{{./Figures/}}

\newtheorem{definition}{Definition}

\newtheorem{theorem}{Theorem}
\newtheorem{proposition}{Proposition}

\newtheorem{lemma}{Lemma}

\renewcommand{\P}{\mathbb{P}}

\newcommand{\Sus}{\mathbb{S}}

\newcommand{\cT}{\mathcal{T}}
\newcommand{\cX}{\mathcal{X}}
\newcommand{\mX}[3]{{X^{#1}_{#2}[#3]}}

\begin{document}
\title{Finding an Infection Source under the SIS Model}

\name{Wuqiong~Luo and Wee Peng Tay}
\address{Nanyang Technological University \\50 Nanyang Avenue, Singapore 639798}


\maketitle

\begin{abstract}
We consider the problem of identifying an infection source based only on an observed set of infected nodes in a network, assuming that the infection process follows a Susceptible-Infected-Susceptible (SIS) model. We derive an estimator based on estimating the most likely infection source associated with the most likely infection path. Simulation results on regular trees suggest that our estimator performs consistently better than the minimum distance centrality based heuristic.
\end{abstract}


\vspace{-5mm}
\section{Introduction}\label{sec:Introduction}
\vspace{-3mm}

We do not have immunity against bacterial diseases like typhoid fever, Methicillin-resistant Staphylococcus aureus (MRSA), and tuberculosis. An infected individual can become infected again with the same disease even after recovering from it. The spread of such diseases are often modeled using a Susceptible-Infected-Susceptible (SIS) model \cite{Hethcote1976,Newman2003}. In a discrete time SIS model, at each time step, the individuals who have the disease are in \textit{infected} state, and those individuals who may potentially get infected at the next time step by currently infected individuals are said to be \textit{susceptible}. An infected individual may recover from the disease and get infected again at subsequent time steps \cite{Allen1994}. A computer virus spreading in a computer network without effective anti-virus counter-measures can also be modeled using a SIS model as a computer that has been cleaned of its infection may get re-infected again \cite{Saito2012}. Opinion dynamics in a social network may also be modeled in some cases using SIS models. A individual on Twitter \cite{Cha2010} may be influenced by the opinion or posting of someone she is following, thereby becoming ``infected'' with the same opinion. She can subsequently change her opinion and become ``uninfected'' again. In all these examples, we may want to identify or estimate a subset of nodes in the network that first started the infection. In the case of a disease, identification of the infection sources or index cases aids epidemiological studies, while tracing of the sources of a computer malware helps to track down the perpetrators.

Existing works related to infection spreading in a network have primarily focused on the parameters of the diffusion process such as the outbreak thresholds and the effect of network structures \cite{Moore2000, Newman2002, Oneill2002, Ganesh2005}. Little work has been done on identifying the infection sources. One of the first works to address the infection source identification problem is \cite{Shah2011}, who consider a Susceptible-Infected (SI) model, where susceptible nodes may get infected, while infected nodes do not recover. A minimum distance centrality (DC) based estimator was proposed to identify the most likely infection source. Subsequently, \cite{Luo2012Secon, Luo2012CoRR} considers the problem of identifying multiple infection sources under the SI model, while \cite{Zhu2012} studies the single infection source identification problem for the Susceptible-Infected-Recovered (SIR) model, where an infected node may recover but can never be infected again. A computationally efficient sample path based estimator was proposed in \cite{Zhu2012} to estimate the infection source. However, as alluded to earlier, the assumption that a recovered node can never be infected again is false in a lot of practical examples.

In this paper, we study the single infection source estimation problem for an SIS model. We assume that we only observe one snapshot of the infection spreading process at some point in time, and derive an estimator that finds the source node associated with the most likely infection process that yields the observed snapshot. The estimator we derive is the same as that in \cite{Zhu2012}, which considers an SIR model, showing that the proposed estimator is relatively robust to the underlying infection and recovery process of the nodes. This is somewhat surprising as the two models are significantly different. We also note that the optimality proofs of our estimator differ significantly from that in \cite{Zhu2012}. Simulation results suggest that our estimator performs better than the minimum distance centrality based estimator \cite{Shah2011}. Our method can also be viewed as a data-driven proxy to finding the most ``influential'' node in an SIS infection network, in contrast to \cite{Saito2012}, which determines the influential nodes based on the expected number of infected nodes.

The rest of this paper is organized as follows. In Section \ref{sec:problem_formulation}, we present the SIS model and problem formulation. In Section \ref{sec:single_source_estimation}, we describe our source estimator, and present simulation results in Section \ref{sec:simulation_results} to evaluate the performance of the proposed estimator on regular trees. Finally we conclude and summarize in Section \ref{sec:conclusion}.

\vspace{-5mm}
\section{Problem Formulation}\label{sec:problem_formulation}
\vspace{-3mm}
Consider an undirected graph $G=(V,E)$, where each node is either infected or uninfected. If a node is infected, we let the state of the node be 1, and 0 otherwise. We assume that time is divided into discrete time slots, and the state of a node $v$ in time slot $t$ is given by $X_v[t]$. At time $t=0$, we assume that there is only one infected node $s^* \in V$, which we call the infection source. At the beginning of a time slot $t$, let the set of all infected nodes and their neighbors be $\Sus(t)$. We call these the susceptible nodes as they may become infected by the end of time slot $t$, while those nodes not in $\Sus(t)$ remains uninfected with probability one. Let $q \in (0,1)$ be the probability that any node $v \in \Sus(t)$ becomes infected at the end of time slot $t$, i.e., $\P(X_v[t+1] = 1) = q$ if $v\in\Sus(t)$. We also assume that the susceptible nodes become infected independently of each other. For any set $J$, let $X_J[t] = \{X_v[t]: v\in J\}$ be the collection of the states of nodes in $J$, and let $X_V [0,t] = \{ X_v[\tau] : 0\leq \tau \leq t, v \in V\}$ denote an infection \textit{path} from time $0$ to $t$.

At some time slot $t$, we observe the set of all infected nodes $V_I$, which we assume to be non-empty. We do not assume that we know the elapsed time $t$. The problem of identifying the infection source can be formulated as a maximum likelihood (ML) estimation problem by treating the infection source $s^*$ and the elapsed time $t$ as parameters to be estimated. We want to identify the node $\hat{s}_{ML} \in V$ and the time $\hat{t}_{ML}$ that maximizes the likelihood of the observed infection set $V_I$, given by
\begin{align*}
(\hat{s}_{ML},\hat{t}_{ML}) =\arg \max_{\substack{v\in V}}\sum_{X_V[0,t] \in \cX_v} \P(X_V[0,t] \mid s^* = v),
\end{align*}
where $\cX_v$ is the set of all possible infection paths starting with $v$ and resulting in $V_I$, and $\P(X_V[0,t] \mid s^* = v)$ is the likelihood of $X_V[0,t]$ given that the infection source is $v$. Unlike the infection sources identifying problem for SI model \cite{Shah2011, Luo2012Secon, Luo2012CoRR}, finding the ML estimator for the SIS model is very challenging as the set of nodes that had been infected before time $t$ is a superset of the observed $V_I$. This implies that, unlike the SI model, the most likely infection source may not be in $V_I$. In the following, we propose an approximation by finding the source node associated with the most likely infection path.

\vspace{-4mm}
\section{Infection path based estimation for regular trees}\label{sec:single_source_estimation}
\vspace{-2mm}

Assume that the underlying network $G$ is an infinite regular tree. We propose as the infection source estimate the node associated with the most likely infection process:
\begin{align}
\hat{s} = \arg \max_{v \in V} \max_{t \in \cT_v, X_V[0,t] \in \cX_v} \P(X_V[0,t] \mid s^* = v),
\label{equ:proposed_single_source_estimator}
\end{align}
where $\cT_v$ is the set of all feasible observation times if $v$ is the infection source. The same estimator has also been used for the SIR model in \cite{Zhu2012}. The estimator $\hat{s}$ in \eqref{equ:proposed_single_source_estimator} can be found in two steps. For each $v \in V$ as the infection source, we determine the most likely infection path induced by $v$ being the infection source. Then, we find $\hat{s}$ as the node that maximizes the likelihood of the most likely infection path found in the first step.

%
\vspace{-3mm}
\subsection{Most likely infection path}
\vspace{-1mm}

We start with two definitions, the first of which is borrowed from \cite{Zhu2012}.

\begin{definition} \label{def:infection_eccentricity}
Let $d(v,u)$ denote the length of the shortest path between $v$ and $u$, which is also called the distance between $v$ and $u$. Define the largest distance between $v$ and any infected node to be,
\begin{align*}
\bar{d}(v,V_I)=\max_{u \in V_I} d(v,u).
\end{align*}
We call $\bar{d}(v,V_I)$ the infection eccentricity of node $v$. Furthermore, the nodes with minimum infection eccentricity are defined as Jordan infection centers of $V_I$.
\end{definition}

\begin{definition} \label{def:most_likey_path_for_t}
For each $v\in V$ and $t \in \cT_v$, let $\mX{v}{V}{0,t} \in \cX_v$ to be the most likely infection path up to time $t$, given that $v$ is the infection source, i.e.,
\begin{align*}
\mX{v}{V}{0,t} = \arg \max_{X_V[0,t] \in \cX_v}\P(X_V[0,t] \mid s^* = v).
\end{align*}
For any set $J$, and $0 \leq i \leq j \leq t$, we let $\mX{v,t}{J}{i,j}$ be the states of nodes in $J$ during time slots $i$ to $j$, in the infection path $\mX{v}{V}{0,t}$.
\end{definition}

We use the following notations throughout this paper.
\begin{itemize}
  \item Given any $v \in V$, let $V_v(h)$ to be the set of nodes $h$ hops away from $v$. \label{def:layered_nodes}
  \item For any tree $G$ and a pair of nodes $u,v \in G$, let $T_u(v;G)$ be the subtree of $G$ rooted at node $u$ with the first link in the path from $u$ to $v$ removed. \label{def:subtree}
\end{itemize}




The following lemma provides an important property, the proof of which is omitted due to space constraints.
\begin{lemma}\label{lemma:better_later}
Suppose that $v\in V$ is the infection source. For an observed set of infected nodes $V_I$, let $H$ be the minimum connected subgraph of $G$ that contains $V_I$ and $v$. Then, for any $t \in \cT_v$, and any $u \in H \backslash \{v\}$, the first infection time $t_u^I$ of $u$ is bounded by
\begin{align}
t_u^I \in [d(v,u), t-\max_{x \in T_u(v;H)}d(u,x)], \label{equ:first_infection_time_bound}
\end{align}
Furthermore, in the most likely infection path $\mX{v,t}{V}{0,t}$ (Definition \ref{def:most_likey_path_for_t}), the first infection time for $u$ is given by
\begin{align}
\tilde{t}_u^I = t-\max_{x \in T_u(v;H)}d(u,x). \label{equ:optimal_first_infection_time}
\end{align}
\end{lemma}

Lemma \ref{lemma:better_later} shows that to find the most likely infection path conditioned on $v$ being the infection source, we should choose the the first infection time for any non-source node to be as late as possible.

\begin{lemma}\label{lemma:optimal_t}
Given a non-empty set of infected nodes $V_I$, suppose that $v$ is the infection source. Then,
\begin{enumerate}[(1)]
  \item \label{lemma:optimal_t_feasible_set} the set of all feasible observation times is $\cT_v=[\bar{d}(v,V_I), \infty)$;
  \item \label{lemma:optimal_t_monotonically_decreasing}$\P(\mX{v}{V}{0,t})$ is monotonically decreasing in $t \in \cT_v$; and
  \item \label{lemma:optimal_t_optimal_t} the most likely elapsed time conditioned on $v$ being the infection source is given by $t^v=\bar{d}(v,V_I)$.
\end{enumerate}
\end{lemma}

\begin{proof}
We first prove claim \eqref{lemma:optimal_t_feasible_set}. The infection can propagate at most one hop further from the source node $v$ in one time slot. If $t<\bar{d}(v,V_I)$, the infection can not reach the nodes $V_v(\bar{d}(v,V_I))$, and therefore, it is not possible for $V_v(\bar{d}(v,V_I))$ to become infected. This proves claim \eqref{lemma:optimal_t_feasible_set}.

Next, we show claim \eqref{lemma:optimal_t_monotonically_decreasing}. Fix a $t \in \cT_v$. We compare $\mX{v}{V}{0,t}$ with $\mX{v}{V}{0,t+1}$. We first show that the source node $v$ is susceptible at time slot 1 in $\mX{v}{V}{0,t+1}$, i.e. $v \in \Sus(1)$. This is true because if $v \notin \Sus(1)$, then $v$ and all of its neighboring nodes are uninfected at time slot 1, i.e., $\mX{v,t+1}{V_v(1) \bigcup \{v\}}{1}=0$. This implies that the set of infected nodes $V_I$ is empty as $v$ is the only source in the network. This contradicts our assumption that at least one node is infected.

Since $\mX{v}{V}{0,t} \in \cX_v$, from Lemma \ref{lemma:better_later}, we have that $\mX{v,t+1}{V}{2,t+1}$ corresponds to $\mX{v,t}{V}{1,t}$ and $\mX{v,t+1}{V_v(1)}{1}=0$, so that $\mX{v,t+1}{v}{1}=1$, yielding
\begin{align}
\frac{\P(\mX{v}{V}{0,t+1})}{\P(\mX{v}{V}{0,t})} = &\P(\mX{v,t+1}{v}{1}=1) \cdot \P(\mX{v,t+1}{V_v(1)}{1}=0) \cdot \nonumber \\
&\frac{\P(\mX{v,t+1}{V}{2,t+1}}{\P(\mX{v,t}{V}{1,t})} \nonumber \\
=& q(1-q)^{|V_v(1)|} < 1, \label{equ:time_diff_1}
\end{align}
where $|V_v(1)|$ denotes the number of elements in the set $V_v(1)$, and $\P(\mX{v,t+1}{V}{2,t+1})=\P(\mX{v,t}{V}{1,t})$ because  $\mX{v,t+1}{V}{1,t+1}=\mX{v,t}{V}{0,t}$. From \eqref{equ:time_diff_1} we can see that $\P(\mX{v}{V}{0,t})$ is monotonically decreasing as $t$ increases, which proves claim \eqref{lemma:optimal_t_monotonically_decreasing}. The last claim now follows from claim \eqref{lemma:optimal_t_monotonically_decreasing}, and the proof for Lemma \ref{lemma:optimal_t} is now complete.
\end{proof}

\vspace{-3mm}
\subsection{Source associated with the most likely infection path}
\begin{proposition}\label{prop:tv=tu-1}
Let $H$ be the minimum connected subgraph of $G$ that contains $V_I$. Suppose that $u$ and $v$ are neighboring nodes in $H$ with $\bar{d}(v,V_I)<\bar{d}(u,V_I)$. Let $l=\arg \max_{x\in V_I}d(u,x)$, then we have
\begin{enumerate}[(1)]
  \item \label{prop:tv=tu-1_leaf}$l \in T_v(u;H)$;
  \item \label{prop:tv=tu-1_tv=tu-1}$t^v=d(v,l)=t^u-1$.
\end{enumerate}
\end{proposition}


\begin{proof}
Note that $t^u=\bar{d}(u,V_I)=d(u,l)$ by Lemma \ref{lemma:optimal_t}(\ref{lemma:optimal_t_optimal_t}). If $l \notin T_v(u; H)$, we have $d(v,l)=d(u,l)+1=t^u+1$. From Lemma \ref{lemma:optimal_t}(\ref{lemma:optimal_t_optimal_t}), we have $t^v =\bar{d}(v,V_I) \geq d(v,l)$, so $t^v \geq t^u + 1$, which contradicts the assumption that $t^v<t^u$. This completes the proof of the first claim, which now implies that $d(v,l)=d(u,l)-1=t^u-1$. From Lemma \ref{lemma:optimal_t}(\ref{lemma:optimal_t_optimal_t}), we obtain $t^v \geq d(v,l)$, so that $t^u-1 \leq t^v < t^u$, which gives us $t^v = t^u-1$. This completes the proof for the proposition.
\end{proof}

\begin{lemma}\label{lemma:better_neighbor}
Let $H$ be the minimum connected subgraph of $G$ that contains $V_I$, and let $t^v=\bar{d}(v,V_I)$ for any $v \in H$. Then, for any pair of neighboring nodes $u$ and $v$ in H with $t^v<t^u$, we have
\begin{align*}
\P(\mX{v}{V}{0,t^v})>\P(\mX{u}{V}{0,t^u}).
\end{align*}
\end{lemma}

\begin{proof}
Denote the first infection time of $v$ in the infection path $\mX{u}{V}{0,t^u}$ as $t_v^I$. We first show that $t_v^I=1$ in the infection path $\mX{u}{V}{0,t^u}$. Conditioned on node $u$ being the infection source, the infection can propagate at most $t^u-t_v^I$ hops away from node $v$ within the subtree $T_v(u;H)$. From Proposition \ref{prop:tv=tu-1}(\ref{prop:tv=tu-1_tv=tu-1}), if $t_v^I>1$, we have $d(v,l)=t^u-1>t^u-t_V^I$, for $l=\arg \max_{x\in V_I}d(u,x)$. In other words, the infection can not reach node $l$, which is a contradiction. Therefore, we must have $t_v^I=1$ in the infection path $\mX{u}{V}{0,t^u}$.

Using the same arguments as in the proof of Lemma \ref{lemma:optimal_t}, one can show that $\mX{u,t^u}{V}{2,t^u}$  corresponds to $\mX{v,t^v}{V}{1,t^v}$, with $ \mX{u,t^u}{V_u(1)\backslash \{v\}}{1}=0$ and $\mX{u,t^u}{u}{1}=0$. We then obtain
\begin{align*}
\frac{\P(\mX{u}{V}{0,t^u})}{\P(\mX{v}{V}{0,t^v})}  = &\frac{\P(\mX{u,t^u}{V}{2,t^u})}{\P(\mX{v,t^v}{V}{1,t^v})} \cdot \P(\mX{u,t^u}{v}{1}=1) \cdot \\
& \P(\mX{u,t^u}{u}{1}=0) \cdot \P(\mX{u,t^u}{V_u(1)\backslash \{v\}}{1}=0) \\
= & q(1-q)^{|V_u(1)|} < 1,
\end{align*}
where $\P(\mX{u,t^u}{V}{2,t^u})=\P(\mX{v,t^v}{V}{1,t^v})$ because $\mX{u,t^u}{V}{1,t^u}\\=\mX{v,t^v}{V}{0,t^v}$. The proof for Lemma \ref{lemma:better_neighbor} is now complete.
\end{proof}


%

We have the following result based on Lemma \ref{lemma:optimal_t} and Lemma \ref{lemma:better_neighbor}.
\begin{theorem}\label{theorem:single_source_estiamte_Jordan_infection_center}
If $V_I$ is an observed infection set in an infinite tree, the estimator in \eqref{equ:proposed_single_source_estimator} is given by
\begin{align}
\hat{s} \in \arg \min_{v \in V} \bar{d}(v,V_I), \label{equ:Jordan_infection_center}
\end{align}
i.e., the infection source associated with the most likely infection path is a Jordan infection center (cf.\ Definition \ref{def:infection_eccentricity}).
\end{theorem}

\begin{proof}
It is easy to see that if $G$ is a tree, then there are at most two Jordan infection centers for $V_I$. In addition, if there are indeed two Jordan infection centers, they are neighboring nodes\cite{Zhu2012}. With out loss of generality, we assume there is only one Jordan infection center $\hat{s}$. (When there are two Jordan infection centers, we can treat them as a single virtual node.) For any node $v_k \in H \backslash \{\hat{s}\}$, we denote the path from $\hat{s}$ to $v_k$ as $[\hat{s}, v_1, v_2, \cdots, v_k]$ , where $k \geq 1$. We want to show that
\begin{align}
\P(\mX{\hat{s}}{V}{0,t^{\hat{s}}}) > \P(\mX{v_k}{V}{0,t^{v_k}}). \label{equ:Jordan_infection_center_optimal}
\end{align}

Fix a $l \in V_I$ such that $d(\hat{s},l)=\bar{d}(\hat{s},V_I)$. Let $u$ denote the neighboring node of $\hat{s}$ on the path from $\hat{s}$ to $l$. Consider a node $l'$, where $l'=\arg \max_{v\in V_I \backslash T_u(\hat{s};H)}d(\hat{s},v)$. We first show that $d(\hat{s},l') \geq d(\hat{s},l)-1$. This is true because if $d(\hat{s},l') \leq d(\hat{s},l)-2$,
\begin{align*}
\bar{d}(u,V_I) &= \max \left(d(u,l'),d(u,l)\right) \\
&=\max \left(d(\hat{s},l')+1,d(\hat{s},l)-1\right) \\
&=d(\hat{s},l)-1.
\end{align*}
In the last line, we have used the inequality that $d(\hat{s},l')+1 \leq d(\hat{s},l)-2 +1 = d(\hat{s},l)-1$. We find a node $u$ that has infection eccentricity less than that of $\hat{s}$, a contradiction. 

Note that $l$ could be either in the subtree $T_{v_1}(\hat{s};H)$ or not, we shall look into these two possible cases. When $l \notin T_{v_1}(\hat{s};H)$, it is easy to see that $\bar{d}(v_i,V_I)=d(\hat{s},l)+i$ for $1 \leq i \leq k$. When $l \in T_{v_1}(\hat{s};H)$, suppose that $\bar{d}(v_{i+1},V_I) \leq \bar{d}(v_{i},V_I)$ for some $i \in [1,k-1]$. Let $\tilde{l}$ be a node such that $d(v_i, \tilde{l}) = \bar{d}(v_{i},V_I)$. Then, we must have $\tilde{l} \in T_{v_{i+1}}(v_i;H)$, otherwise we have a contradiction. We then have
\begin{align*}
\bar{d}(v_{i+1},V_I)
&\geq d(v_{i+1},l') \\
&= d(\hat{s},l') + i + 1\\
&\geq i+d(\hat{s},l)\\
&\geq 2i + d(v_i,\tilde{l}) \\
& > d(v_i,\tilde{l}),
\end{align*}
a contradiction. Therefore, we have $ \bar{d}(v_{i},V_I)<\bar{d}(v_{i+1},V_I)$ for $i \in [1,k-1]$. Furthermore, we have $ \bar{d}(\hat{s},V_I)<\bar{d}(v_{1},V_I)$ by assumption. By repeatedly applying Lemma \ref{lemma:better_neighbor} and Lemma \ref{lemma:optimal_t}(\ref{lemma:optimal_t_optimal_t}), we can show \eqref{equ:Jordan_infection_center_optimal} is true for both cases, and the proof for Theorem \ref{theorem:single_source_estiamte_Jordan_infection_center} is now complete.
\end{proof}
\vspace{-1mm}
Theorem \ref{theorem:single_source_estiamte_Jordan_infection_center} shows that the optimal estimator in \eqref{equ:proposed_single_source_estimator} is given by a Jordan infection center. Note that this is the same result if an SIR infection process is assumed \cite{Zhu2012}. We have therefore shown that using Jordan infection centers is relatively robust to the underlying assumptions governing the infection and recovery of nodes in the network.

\vspace{-4mm}
\section{Simulation Results}\label{sec:simulation_results}
\vspace{-2mm}

In this section, we present simulation results on regular trees to evaluate the performance of the proposed estimator. An efficient algorithm has been described in \cite{Zhu2012} to find the Jordan infection center, which we call optimal infection path (OIP) algorithm. We refer the reader to \cite{Zhu2012} for details of OIP. The benchmark is the minimum DC based estimator that is proved in \cite{Shah2011} to be the maximum likelihood infection source estimator for regular trees in SI models.

In each simulation run, we let the underlying network $G$ be a sufficiently large regular tree, so that $G$ can be treated as an infinite tree. We set the degree of the regular trees to be $2,3,4,5$ or $6$. For each degree, we perform 1000 simulation runs. We randomly choose a node as the infection source and let the infection spread out using the SIS model. The infection probability $q$ is chosen uniformly from $(0,1)$. We observe the infection graph after $t$ time slots, where $t$ is chosen uniformly from $[3,5]$. We run the OIP algorithm and the DC algorithm on the observed graph for the proposed estimator and the DC based estimator respectively. Figure \ref{fig:regular_tree_detection_rate} shows the detection rate (percentage of times that the estimator correctly finds the infection source) of both estimators. We can see that the proposed estimator has higher detection rate than the DC based estimator for all kinds of regular trees. Error distance is defined as the distance between the estimate and the infection source. Figure \ref{fig:error_distance_histogram_regular_tree_degree4} shows the histogram of the error distances of both estimators for regular trees with degree 4 (similar results are obtained for other degrees). We see that the proposed estimator has smaller error distance on average.
\vspace{-.3cm}
\begin{figure}[!ht] 
  \centering
  \includegraphics[width=0.42\textwidth]{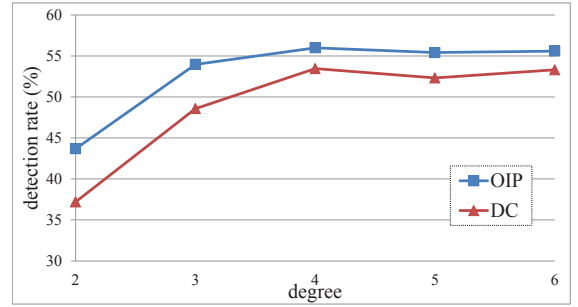}
  \caption{Detection rate of optimal infection path (OIP) based estimator and distance centrality (DC) based estimator for regular trees with various degrees.}
  \label{fig:regular_tree_detection_rate}
\end{figure}
\vspace{-.6cm}
\begin{figure}[!ht] 
  \centering
  \includegraphics[width=0.42\textwidth]{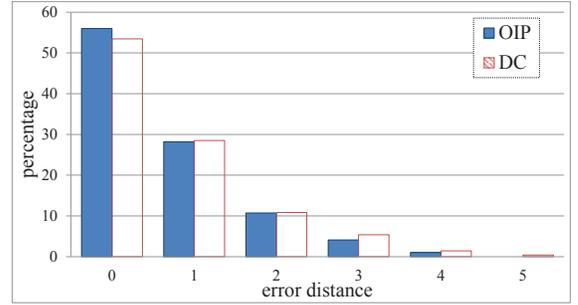}
  \caption{Histogram of error distances of optimal infection path (OIP) based estimator and distance centrality (DC) based estimator for regular trees with degree 4.}
  \label{fig:error_distance_histogram_regular_tree_degree4}
\end{figure}

\vspace{-4mm}
\section{Conclusion}\label{sec:conclusion}
\vspace{-2mm}

We have derived an infection source estimator for an SIS model that identifies the node associated with the most likely infection path. We showed that the estimator is a Jordan infection center. Simulation results on regular trees indicate that our estimator outperforms the minimum distance centrality based estimator, which is proved to be the maximum likelihood estimator for the SI model. In this paper, we make the assumption that there is only one infection source in an infinite regular tree. However, there may exist multiple infection sources in practical applications \cite{Luo2012Secon, Luo2012CoRR}. Future work includes identifying multiple infection sources for the SIS model in a general network.

\vfill\pagebreak
\bibliography{IEEEabrv,SIS}{}
\bibliographystyle{IEEEtran}

\end{document}